\documentclass{amsart}
\usepackage{graphicx}
\usepackage[numbers]{natbib}
\usepackage{amssymb}
\usepackage{amsthm}
\usepackage{amsmath}
\usepackage{mathtools}
\newtheorem{theorem}{Theorem}
\newtheorem{corollary}{Corollary}
\newtheorem{lemma}{Lemma}
\newtheorem{example}{Example}
\newtheorem{definition}{Definition}
\newtheorem{axiom}{Axiom}
\usepackage{multicol}

\usepackage[pdftitle={Evaluating betting odds and free coupons using desirability},%
pdfauthor={Nawapon Nakharutai, Camila C. S. Caiado, Matthias C. M. Troffaes},%
pdfkeywords={betting; coupon; Choquet integration; complementary slackness}]{hyperref}
\usepackage[nosort]{cleveref}
\usepackage{doi}
\usepackage{tikz}
\usetikzlibrary{shapes.geometric, arrows}
\usetikzlibrary{er,positioning}
\usetikzlibrary{shapes}
\usepackage{subcaption}
\usepackage[font=small,skip=5pt]{caption}

\usepackage{multirow}
\usepackage{lscape}
\usepackage{adjustbox}
\newcommand*\rot{\rotatebox{90}}
\usepackage{algorithm}
\usepackage{algorithmic}

\newlength\myindent
\makeatletter
\newenvironment{breakablealgorithm}
  {
   \begin{center}
     \refstepcounter{algorithm}
     \hrule height.8pt depth0pt \kern2pt
     \renewcommand{\caption}[2][\relax]{
       {\raggedright\textbf{\ALG@name~\thealgorithm} ##2\par}%
       \ifx\relax##1\relax 
         \addcontentsline{loa}{algorithm}{\protect\numberline{\thealgorithm}##2}%
       \else 
         \addcontentsline{loa}{algorithm}{\protect\numberline{\thealgorithm}##1}%
       \fi
       \kern2pt\hrule\kern2pt
     }
  }{
     \kern2pt\hrule\relax
   \end{center}
  }
\makeatother

\DeclareMathOperator{\dom}{dom}

\newcommand{\linprogref}[1]{\textup{(#1)}}

\title{Evaluating betting odds and free coupons using desirability}

\author{Nawapon Nakharutai}
\address{Durham University, Department of Mathematical Sciences, UK}
\email{nawapon.nakharutai@durham.ac.uk}

\author{Camila C. S. Caiado}
\address{Durham University, Department of Mathematical Sciences, UK}
\email{c.c.d.s.caiado@durham.ac.uk}

\author{Matthias C. M. Troffaes}
\address{Durham University, Department of Mathematical Sciences, UK}
\email{matthias.troffaes@durham.ac.uk}


\begin{document}

\begin{abstract}
In the UK betting market, bookmakers often offer a free coupon to new customers. These free coupons allow the customer to place extra bets, at lower risk, in combination with the usual betting odds. We are interested in whether a customer can exploit these free coupons in order to make a sure gain, and if so, how the customer can achieve this. To answer this question, we evaluate the odds and free coupons as a set of desirable gambles for the bookmaker.

We show that we can use the Choquet integral to check whether this set of desirable gambles incurs sure loss for the bookmaker, and hence, results in a sure gain for the customer. In the latter case, we also show how a customer can determine the combination of bets that make the best possible gain, based on complementary slackness.

As an illustration, we look at some actual betting odds in the market and find that, without free coupons, the set of desirable gambles derived from those odds avoids sure loss. However, with free coupons, we identify some combinations of bets that customers could place in order to make a guaranteed gain.
\end{abstract}

\maketitle

\section{Introduction}\label{intro}
Consider the football betting market in the UK where a bookmaker typically offers fractional betting odds for possible outcomes. For example, in a match between Manchester United and Liverpool, the bookmaker offers odds in the form $a/b$ for Manchester United winning, $c/d$ for a draw and $e/f$ for Liverpool winning.
Suppose a customer accepts the odds $a/b$ by placing a stake of $b$ pounds on a Manchester United win, which he pays to the bookmaker in advance of the match. After the match, if Manchester United wins, the bookmaker will pay him $a+b$ pounds. So, if Manchester United wins, then the customer's total return will be $a$ pounds; otherwise the customer will lose $b$ pounds. 

To predict the outcome of a match, the bookmaker may encounter difficulties such as lack of data (e.g. team A has never played with team B during last five years), missing data, limited football expert opinion, or even contradicting information from different football experts. Various authors \cite{1975:williams:condprev,2007:williams:condprev,1991:walley,2014:troffaes:decooman::lower:previsions} have argued that these issues can be handled by using \emph{sets of desirable gambles}. A gamble represents a reward (i.e. money in our case) that depends on an uncertain outcome (i.e. the match result). The bookmaker can model his belief about this outcome by stating a collection of gambles that he is willing to offer. Such set is called a set of desirable gambles. Through duality, stating a set of desirable gambles is mathematically equivalent to stating a set of probability distributions.

If there are no combinations of desirable gambles that result in a guaranteed loss, then we say that a set of desirable gambles \emph{avoids sure loss} \cite{1975:williams:condprev,2007:williams:condprev}. Thus, if the bookmaker's set of desirable gambles avoids sure loss, then there is no combination of bets from which customers can make a guaranteed gain. On the other hand, if the set does not avoid sure loss, then there is a combination of bets that customers can exploit to incur a sure gain.

In addition to avoiding sure loss, the bookmakers also want to entice new customers. There are several techniques that bookmakers can use to persuade customers to bet with their companies. Some bookmakers may offer greater betting odds than others since greater odds means a greater payoff to the customers. Another technique is to offer a ``free coupon", which is a stake that customers can spend on betting. The free coupon can also be viewed as part of a desirable gamble.

However, bookmakers may worry that customers will find a combination of different odds and free coupons that they can bet on and make a guaranteed profit. Therefore, from the bookmaker's perspective, they would like to check whether sets of desirable gambles derived from different odds and free coupons avoid sure loss or not.
Conversely, in theory, a customer may be interested in the case where the bookmaker's set does not avoid sure loss, because then the customer can make a guaranteed profit. In that case, a customer may want to find the combination of bets which results in the best possible sure gain.

There are several studies on exploiting betting odds and free bets in order to find strategies that make a profit. For example,
\citet[Appendix I]{1991:walley} and \citet{2017:Quaeghebeur:Wesseling:Beauxis-Aussalet} study an application of sets of desirable gambles on sports; \citet{2009:Milliner:White:Webber,1998:Schervish:Seidenfeld:Kadane,vlastakisbeating} exploit betting odds directly, whilst \citet{2013:Emiliano} takes free bets into account. Emiliano considers the case of only two possible outcomes, and allows cooperation between customers.
In this paper, we look at any finite number of possible outcomes, but we only consider a single customer.
We evaluate betting odds and free coupons and check whether a set of desirable gambles derived from odds and free coupons avoids sure loss (or not) via the natural extension. If the set does not avoid sure loss, then we show exactly how a customer can incur a sure gain. 

In general, one can check avoiding sure loss by solving a linear programming problem  \cite[p.~151]{1991:walley}. In our previous work \citep{2018:Nakharutai:Troffaes:Caiado}, we provided efficient algorithms for solving these linear programming problems.
For our specific problem, we show that we can calculate the natural extension through the Choquet integral, or through solving a linear programming problem where the optimal value is equal to the natural extension. In the case of not avoiding sure loss, we know that we can find a strategy that the customer can bet on to make a guaranteed gain. We show that this strategy can be identified using the Choquet integral and complementary slackness conditions.
Our method for finding this strategy is generally applicable not just to this betting problem, but to arbitrary problems involving upper probability mass functions. Specifically, by using the Choquet integral and exploiting complementary slackness conditions, we can find optimal solutions of the corresponding pair of dual linear programming programs without directly solving them. 

The paper is organised as follows. \Cref{sec:asl&ne} briefly reviews the main concepts behind desirability, avoiding sure loss and natural extension. We also discuss the Choquet integral which can be used to calculate the natural extension. In \cref{sec:bet}, we introduce fractional fixed odds and explain how betting odds work. As betting odds can be viewed as a set of desirable gambles, we revisit a simple known algorithm to check whether such set avoids sure loss or not.
In \cref{sec:free bet}, we discuss free coupons from the perspective of desirability. We show how we can check whether the problem with free coupons avoids sure loss or not, by means of the natural extension. We demonstrate how we can use the Choquet integral to calculate this natural extension. Next, we exploit complementary slackness to find a combination of bets which makes the best possible guaranteed gain. To illustrate our results, in \cref{sec:Euro-odds}, we consider some actual betting odds and free coupons in the market, and provide an example where a customer can make a sure gain with a free coupon.
\Cref{sec:conclusion} concludes this paper.

\section{Avoiding sure loss and natural extension}\label{sec:asl&ne}

In this section, we will briefly discuss desirability, avoiding sure loss and natural extension.
We will also explain the Choquet integral which can be used to calculate the natural extension in the case considered in this paper. The material in this section will be useful later when we view betting odds and free coupons as a set of desirable gambles and when we want to check whether this set avoids sure loss or not.

\subsection{Avoiding sure loss}\label{sec:asl}
Let $\Omega$ be a finite set of uncertain outcomes. A \emph{gamble} is a bounded real-valued function on $\Omega$. Let $\mathcal{L}(\Omega)$ denote the set of all gambles on $\Omega$. Let $\mathcal{D} $ be a finite set of gambles that a subject deems acceptable; we call $\mathcal{D}$ the subject's \emph{set of desirable gambles}. Rationality conditions for desirability have been proposed as follows \cite[p.~29]{2014:troffaes:decooman::lower:previsions}:

\begin{axiom}[Rationality axioms for desirability]
For every $f$ and $g$ in $\mathcal{L}(\Omega)$ and every non-negative $\alpha \in \mathbb{R}$, we have that:
\begin{itemize}
\item[(D1)] If $f\le 0$ and $f\neq 0$, then $f$ is not desirable.
\item[(D2)] If $f\geq 0$, then $f$ is desirable.
\item[(D3)] If $f$ is desirable, then so is $\alpha f$.
\item[(D4)] If $f$ and $g$ are desirable, then so is $f+g$.
\end{itemize}
\end{axiom}
The first two axioms are trivial as the subject should accept any gamble that he cannot lose from, but he should not accept any gamble that he cannot win from. Axiom (D3) follows the linearity of the utility scale and axiom (D4) shows that a combination of desirable gambles should also be desirable.

We do not assume that any set $\mathcal{D}$, specified by the subject,
satisfies these axioms. However, we can use these axioms to examine
the rationality of $\mathcal{D}$. Indeed,
the rationality axioms essentially state that a non-negative combination of desirable gambles should not produce a sure loss \cite[p.~30]{2014:troffaes:decooman::lower:previsions}.
In that case, we say that $\mathcal{D}$ avoids sure loss.

\begin{definition}\cite[p.~32]{2014:troffaes:decooman::lower:previsions}\label{def:asl}
  A set $\mathcal{D}\subseteq\mathcal{L}(\Omega)$ is said to \emph{avoid sure loss} if for
  all $ n \in \mathbb{N}$, all  $\lambda_{1}, \dots,\lambda_{n} \geq 0$, and all $f_{1}, \dots,f_{n} \in \mathcal{D}$,
\begin{equation}\label{eq:asl}
\max_{\omega \in \Omega}\left(\sum_{i=1}^{n} \lambda_{i}f_{i}(\omega)\right)  \geq 0.
\end{equation}
\end{definition}
Note that the rationality axioms for desirability are stronger than the condition of avoiding sure loss \cite[p.~32]{2014:troffaes:decooman::lower:previsions}.

We can also model uncertainty via acceptable buying (or selling) prices for gambles. A \emph{lower prevision} $\underline{P}$ is a real-valued function defined on some subset of $\mathcal{L}(\Omega)$. We denote the domain of $\underline{P}$ by $\dom\underline{P}$. Given a gamble $f\in\dom\underline{P}$, we interpret $\underline{P}(f)$ as a subject's supremum buying price for $f$, i.e. $f-\alpha$ is deemed desirable for all $\alpha<\underline{P}(f)$ \citep[p.~40]{2014:troffaes:decooman::lower:previsions}.

\begin{definition}\cite[p.~42]{2014:troffaes:decooman::lower:previsions}\label{def:47}
  A lower prevision $\underline{P}$ is said to \emph{avoid sure loss} if for
  all $ n \in \mathbb{N}$, all  $\lambda_{1}, \dots,\lambda_{n} \geq 0$, and all $f_{1}, \dots,f_{n} \in \dom\underline{P}$,
\begin{equation}\label{eq:3.1}
\max_{\omega\in \Omega} \left( \sum_{i=1}^{n} \lambda_{i}\left[f_{i}(\omega)-\underline{P}(f_{i})\right] \right) \geq 0. 
\end{equation}  
\end{definition}

Any lower prevision $\underline{P}$ induces a \emph{conjugate upper prevision} $\overline{P}$ on $-\dom \underline{P}\coloneqq\{-f\colon f\in \dom \underline{P}\}$, defined by $\overline{P}(f)\coloneqq -\underline{P}(-f)$ for all $f\in -\dom\underline{P}$ . $\overline{P}(f)$ represents a subject's infimum selling price for $f$ \cite[p.~41]{2014:troffaes:decooman::lower:previsions}.

Next, let $A$ denote a subset of $\Omega$, also called an \emph{event}. Its associated \emph{indicator} function $I_{A}$ is given  by 
\begin{equation}
\forall \omega \in \Omega\colon I_{A}(\omega) \coloneqq \begin{cases}1 & \text{if }\omega \in A \\ 0 & \text{otherwise}.
\end{cases}
\end{equation}

Further in the paper, we will also extensively use \emph{upper probability mass functions}. An upper probability mass function $\overline{p}$ is a mapping from $\Omega$ to $[0,1]$, and represents the following lower prevision \citep[p.~123]{2014:troffaes:decooman::lower:previsions}:
\begin{equation}\label{eq:lp-upper-prob}
\forall \omega \in \Omega\colon \underline{P}_{\overline{p}}(-I_{\{\omega\}})\coloneqq -\overline{p}(\omega),
\end{equation}
where $\dom \underline{P}_{\overline{p}}= \bigcup_{\omega \in \Omega} \{-I_{\{\omega\}}\}$. 
We can check whether $\underline{P}_{\overline{p}}$ avoids sure loss by \cref{thm:asl:up-prob}.
\begin{theorem}\label{thm:asl:up-prob}
\citep[p.~124]{2014:troffaes:decooman::lower:previsions} $\underline{P}_{\overline{p}}$ avoids sure loss if and only if $\sum_{\omega \in \Omega}\overline{p}(\omega) \geq 1$.
\end{theorem}
\begin{proof}
See \citep[p.~124, Prop.~7.2]{2014:troffaes:decooman::lower:previsions} with lower probability mass function $\underline{p} = 0$.
\end{proof}
We can interpret an upper probability mass function as providing an upper bound on the probability of each $\{\omega\}$, for all $\omega\in\Omega$ \citep[p.~123]{2014:troffaes:decooman::lower:previsions}.

\subsection{Natural extension}\label{natural}
The natural extension of a set of desirable gambles $\mathcal{D}$ is defined as the smallest set of gambles which includes all finite non-negative combinations of gambles in $\mathcal{D}$ and all non-negative gambles \cite[\S~3.7]{2014:troffaes:decooman::lower:previsions}:

\begin{definition}\cite[p.~32]{2014:troffaes:decooman::lower:previsions}\label{def:ne}
The \emph{natural extension} of a set $\mathcal{D}\subseteq\mathcal{L}(\Omega)$ is:
\begin{equation}\label{eq:ne}
\mathcal{E}_{\mathcal{D}}\coloneqq\left\lbrace g_{0} + \sum_{i=1}^{n} \lambda_{i}g_{i} \colon g_{0} \geq 0,\,n \in \mathbb{N},\, g_1,\dots,g_n \in \mathcal{D},\,\lambda_1,\dots,\lambda_n \geq 0 \right\rbrace.
\end{equation}
\end{definition}

From this natural extension, we can derive a supremum buying price for any gamble $f$.
\begin{definition}\cite[p.~46]{2014:troffaes:decooman::lower:previsions}\label{def:ne-d}
For any set $\mathcal{D}\subseteq\mathcal{L}(\Omega)$ and $f\in\mathcal{L}(\Omega)$, we define:
\begin{align}\label{eq:6.11}
\underline{E}_{\mathcal{D}}(f) & \coloneqq \sup \left\lbrace\alpha \in \mathbb{R}\colon f-\alpha \in \mathcal{E}_{\mathcal{D}}\right\rbrace\\
& =  \sup \left\lbrace\alpha \in \mathbb{R}\colon f -\alpha \geq \sum_{i=1}^{n} \lambda_{i}f_{i}, n \in \mathbb{N}, f_{i} \in \mathcal{D}, \lambda_{i} \geq 0\right\rbrace.
\end{align}
\end{definition}
Note that $\underline{E}_{\mathcal{D}}$ is finite, and hence, is a lower prevision, if and only if $\mathcal{D}$ avoids sure loss
\cite[p.~68]{2014:troffaes:decooman::lower:previsions}.

We denote the conjugate of $\underline{E}_{\mathcal{D}}$ by $\overline{E}_{\mathcal{D}}$ which is defined by
\begin{equation}
  \overline{E}_{\mathcal{D}}(f)\coloneqq -\underline{E}_{\mathcal{D}}(-f)
  =
  \inf \left\lbrace\beta \in \mathbb{R}\colon \beta - f \geq \sum_{i=1}^{n} \lambda_{i}f_{i}, n \in \mathbb{N}, f_{i} \in \mathcal{D}, \lambda_{i} \geq 0\right\rbrace.
\end{equation}
for all $f$ in $\mathcal{L}(\Omega)$ \citep[p.~124]{1991:walley}.
$\underline{E}_{\mathcal{D}}$ is simply denoted by $\underline{E}$ when there is no confusion.

Given a lower prevision $\underline{P}$, we can derive a set of desirable gambles corresponding to $\underline{P}$ as follows \citep[p.~42]{2014:troffaes:decooman::lower:previsions}:
\begin{equation}\label{eq:5.10}
\mathcal{D}_{\underline{P}} \coloneqq \left\{ g-\mu\colon g \in \dom\underline{P} \text{ and } \mu < \underline{P}(g) \right\}.
\end{equation} 
Combining \cref{def:ne-d} and \cref{eq:5.10} together, we can define the natural extension of $\underline{P}$:
\begin{definition}\cite[p.~47]{2014:troffaes:decooman::lower:previsions}\label{def:E_P}
Let $\underline{P}$ be a lower prevision. The natural extension of $\underline{P}$ is defined for all $f \in \mathcal{L}(\Omega)$ by:
\begin{multline}
\underline{E}_{\underline{P}}(f) \coloneqq \underline{E}_{\mathcal{D}_{\underline{P}}}(f)
\\
=  \sup \left\lbrace\alpha \in \mathbb{R}\colon f -\alpha \geq \sum_{i=1}^{n} \lambda_{i}(f_{i}-\underline{P}(f_{i})) , n\in\mathbb{N},\,f_{i} \in  \dom\underline{P},\,\lambda_{i} \geq 0\right\rbrace. 
\end{multline}
\end{definition} 
Similarly, $\underline{E}_{\underline{P}}$ is finite if and only if $\underline{P}$ avoids sure loss \cite[p.~68]{2014:troffaes:decooman::lower:previsions}.

In the next section, we briefly explain the use of the Choquet integral to calculate the natural extension for the type of lower previsions considered in this paper; see \citep{2014:troffaes:itip:computation,2014:troffaes:decooman::lower:previsions} for more detail.

\subsection{Upper probability mass functions and Choquet integration}\label{Choq-int}

Let $\underline{E}_{\overline{p}}$ be the natural extension of  $\underline{P}_{\overline{p}}$ that avoids sure loss.
Then $\underline{E}_{\overline{p}}$ is 2-monotone and can be computed via the Choquet integral \citep[p.~125]{2014:troffaes:decooman::lower:previsions}.
In this section, based on the results from \citep[Sec.~7.1]{2014:troffaes:decooman::lower:previsions}, we give a closed form expression for this integral.

For simplicity, we denote the natural extension $\underline{E}_{\overline{p}}(I_{A})$ of an indicator $I_{A}$ as $\underline{E}_{\overline{p}}(A)$.  We can use the following theorem to calculate $\underline{E}_{\overline{p}}(A)$. 
\begin{theorem}\label{thm:ne:up-prob}
\citep[p.~125]{2014:troffaes:decooman::lower:previsions} Let $\underline{P}_{\overline{p}}$ avoid sure loss. Then for all $ A \subseteq \Omega$, 
\begin{equation}\label{eq:ne-upper-prob}
\underline{E}_{\overline{p}}(A) = \max\{0, 1-U(A^{\mathsf{c}})\} \qquad \text{and} \qquad \overline{E}_{\overline{p}}(A) = \min \{U(A), 1\},
\end{equation}
where $U(A)\coloneqq\sum_{\omega \in A}\overline{p}(\omega).$ 
\end{theorem}
\begin{proof}
See \citep[p.~125]{2014:troffaes:decooman::lower:previsions} with lower probability mass function $\underline{p} = 0$.
\end{proof}

\begin{theorem}\label{thm:decompose}
Let $f$ be decomposed in terms of its level sets $A_i$, $i=0,1,\dots, n$:
\begin{equation}\label{eq:decompose-f}
f = \sum_{i=0}^{n} \lambda_{i} I_{A_{i}}
\end{equation}
where $\lambda_{0} \in \mathbb{R}$, $\lambda_{1},\dots,\lambda_{n} > 0$ and $\Omega = A_{0} \supsetneq A_{1} \supsetneq \dots \supsetneq A_{n}\neq \emptyset$. Then
\begin{equation}\label{eq:ne:gam}
\underline{E}_{\overline{p}}(f) =  \sum_{i=0}^{n} \lambda_{i}\underline{E}_{\overline{p}}(A_{i}).
\end{equation} 
\end{theorem}
\begin{proof}
The right hand side is the Choquet integral \citep[p.~379, Eq.~(C.8)]{2014:troffaes:decooman::lower:previsions} and the natural extension $\underline{E}_{\overline{p}}(f)$ is equal to the Choquet integral \citep[p.~125, Prop.~7.3(ii)]{2014:troffaes:decooman::lower:previsions} (with lower probability mass function $\underline{p}=0$).
\end{proof}
Note that \cref{thm:decompose} also holds for the upper natural extension.
\begin{corollary}\label{cor:decompose}
Let $f$ be a gamble decomposed as in \cref{eq:decompose-f}. Then 
\begin{equation}\label{eq:cor:upE(f)}
\overline{E}_{\overline{p}}(f) =  \sum_{i=0}^{n} \lambda_{i}\overline{E}_{\overline{p}}(A_{i}).
\end{equation}
\end{corollary}
\begin{proof}
See \cref{app:A}.
\end{proof}
The Choquet integral will be useful when we want to calculate the natural extension later in \cref{sec:free bet}.

\subsection{Avoiding sure loss with one extra gamble}\label{extra-gam}
Let $\mathcal{D}=\{g_1,\dots,g_n\}$ be a set of desirable gambles that avoids sure loss and let $f$ be another desirable gamble. We want to check whether $ \mathcal{D} \cup \{f\}$ still avoids sure loss or not. This idea will be used when we want to check avoiding sure loss with a free coupon in \cref{sec:free bet}.

By the condition of avoiding sure loss in \cref{def:asl}, $\mathcal{D}\cup \{f\}$ avoids sure loss if and only if for all $\lambda_{0} \geq 0$, $ n \in \mathbb{N}$, $g_{i} \in \mathcal{D}$ and $\lambda_{1}, \dots,\lambda_{n} \geq 0$,
\begin{equation}\label{eq:asl:condi}
\max_{\omega\in \Omega} \left(\sum_{i=1}^{n} \lambda_{i}g_{i}(\omega) +\lambda_{0}f(\omega) \right) \geq 0.
\end{equation}
We can simplify \cref{eq:asl:condi} as follows.
\begin{lemma}\label{lem:asl:condi}
Let $\Omega$ be a finite set, $\mathcal{D}= \{g_{1}, \dots,g_{n}\}$ be a set of desirable gambles that avoids sure loss and $f$ be another desirable gamble. Then, $\mathcal{D}\cup \{f\}$ avoids sure loss if and only if for all $n \in \mathbb{N}$, $g_{i} \in \mathcal{D}$ and $\lambda_{1}, \dots,\lambda_{n} \geq 0$,
\begin{equation}\label{eq:asl:condi:2}
\max_{\omega\in \Omega} \left(\sum_{i=1}^{n} \lambda_{i}g_{i}(\omega) +f(\omega) \right) \geq 0.
\end{equation}
\end{lemma}
\begin{proof}
If $\lambda_0 = 0$ in \cref{eq:asl:condi}, then \cref{eq:asl:condi} is trivially satisfied because $\mathcal{D}$ avoids sure loss.
Otherwise $\lambda_0 > 0$, and for all $i$, $\lambda_i \geq 0$, so $\lambda_i/\lambda_0 \geq 0$. Therefore \cref{eq:asl:condi} is equivalent to
\begin{equation}\label{eq:asl:condi:3}
\max_{\omega\in \Omega} \left(\sum_{i=1}^{n} \left(\frac{\lambda_i}{\lambda_0}\right) g_{i}(\omega) +f(\omega) \right) \geq 0.
\end{equation}
Therefore, $\mathcal{D}\cup \{f\}$ avoids sure loss if and only if
 \cref{eq:asl:condi:2} holds. 
\end{proof}

Next, we give a method not only for checking avoiding sure loss of $\mathcal{D}\cup \{f\}$, but also for bounding the worst case loss, which will be useful later in \cref{sec:free bet}.

\begin{theorem}\label{thm:check-asl}
Let $ f \in \mathcal{L}(\Omega)$ and let $\mathcal{D} = \{g_{1}, \dots,g_{n}\}$ be a set of desirable gambles that avoids sure loss. Then, $\mathcal{D}\cup \{f\}$ avoids sure loss if and only if $\overline{E}_{\mathcal{D}}(f) \geq 0$. If $\mathcal{D}\cup \{f\}$ does not avoid sure loss, then there exist $\lambda_1 \geq 0$, \dots, $\lambda_n\ge 0$ such that $f+ \sum_{i=1}^{n} \lambda_{i}g_{i}$, which is a  combination of desirable gambles, results in a loss at least $|\overline{E}_{\mathcal{D}}(f)|$.
\end{theorem}
\begin{proof}
See \cref{app:B}.
\end{proof}
Note that by \cref{def:E_P}, \cref{thm:check-asl} can also be applied to $\overline{E}_{\underline{P}}$.

\section{Betting scheme}\label{sec:bet}

In this section, we explain how fractional betting odds work and look at two scenarios: (i) a customer bets against a bookmaker and (ii) a customer
bets against multiple bookmakers. In both cases, we view betting odds as a set of desirable gambles and check whether such a set avoids sure loss or not.

\subsection{Betting with one bookmaker}\label{one bookie}

In the UK, a bookmaker usually offers fixed fractional odds on possible outcomes of an event that customers are interested in. For example, in the European Football Championship 2016, customers are interested in the winner of the championship. Suppose that a bookmaker sets odds on France, say $9/2$, and one customer accepts this odds. For every stake \pounds 2 that the customer bets on France, he will win \pounds 9 plus the return of his stake. So the bookmaker will lose \pounds 9 in total. Otherwise, the bookmaker will pay nothing and keep \pounds 2. The bookmaker often writes $a/1$ as $a$.

Given fractional odds $a/b$, a customer can simply calculate his return as follows. For every amount $b$ that the customer bets, he will either get nothing (in case the bet is lost), or gain $a$ plus the return of his stake (in case the bet is won). As the bookmaker accepts this transaction, the total payoff can be seen as a desirable gamble, say $g$, to the bookmaker:
\begin{equation}\label{eq:odds-gam1}
g(\omega) = \begin{cases}-a \qquad \text{if }\omega = x \\ \ b \qquad \text{ otherwise}.
\end{cases}
\end{equation}
Note that $-g$ is a desirable gamble to the customer, should the customer decide to accept the bookmaker's odds.

Let $\Omega = \{\omega_1,\dots,\omega_n\}$ be a finite set of outcomes. Suppose that for each $i$, the bookmaker sets betting odds $a_{i}/b_{i}$ on $\omega_{i}$. By \cref{eq:odds-gam1}, these odds can be viewed as a set of desirable gambles $\mathcal{D}= \{g_1,\dots,g_n\}$, where
\begin{equation}\label{eq:odds-gam2}
  g_{i}(\omega) \coloneqq
  \begin{cases}
    -a_i & \text{if }\omega = \omega_i \\
    b_i  & \text{ otherwise}.
\end{cases}
\end{equation}
Given odds $a_{i}/b_{i}$ on $\omega_{i}$, suppose that we modify the denominator in this odds to be $b_{j}$. To do so,  we can multiply $a_{i}/b_{i}$ by  $b_{j}/b_{j}$ to be 
\begin{equation}\label{eq:modify-odds}
a_{i}b_{j}/ b_{i}b_{j} = \left(\frac{ a_{i}b_{j}}{ b_{i}}\right)/b_{j}.
\end{equation}
Are new odds still desirable?
By the rationality axioms for desirability, the modified odds are still desirable.
\begin{lemma}\label{lem:scale-odds}
Let $a/b$ be odds on an outcome $\tilde{\omega}$ that are desirable. Then, for all $\alpha >0$, the odds $\alpha a/\alpha b$ on $\tilde{\omega}$ are also desirable.
\end{lemma}
\begin{proof}
Consider
the desirable gamble corresponding to the odds $a/b$:
\begin{equation}\label{eq:lem:scale-odd:1}
g(\omega)\coloneqq
\begin{cases}
-a \qquad \text{if }\omega = \tilde{\omega} \\ ~~ b\qquad \text{otherwise}.
\end{cases}
\end{equation} 
By rationality axiom (D3), for any $\alpha >0$, the gamble $\alpha g$ is also desirable. Hence, the corresponding odds $\alpha a/\alpha b$ are  also desirable.
\end{proof}
\Cref{lem:scale-odds} will be very useful when we want to modify odds to have the same denominator.

Suppose that the bookmaker specifies betting odds for all possible outcomes in $\Omega$. Before announcing these odds, the bookmaker may want to check whether there is a combination of bets from which the customer can make a sure gain, or in other words, whether he avoids sure loss \citep[Appendix 1, I4, p.~635]{1991:walley}:

\begin{theorem}\label{thm:asl-one-booker}
Let $ \Omega = \{\omega_1,\dots,\omega_n\}$. Suppose $a_{i}/b_{i}$ are betting odds on $\omega_{i}$. For each $i\in\{1,\dots,n\}$, let
\begin{equation}\label{eq:7.7}
 g_{i}(\omega)\coloneqq
\begin{cases}
  -a_{i} & \text{if } \omega = \omega_{i} \\
  b_{i} & \text{otherwise}
\end{cases}
\end{equation} 
be the gamble corresponding to the odds $a_{i}/b_{i}$. Then $\mathcal{D}\coloneqq\{g_1,\dots,g_n\}$ avoids sure loss if and only if 
 \begin{equation}\label{eq:7.8}
 \sum_{i=1}^{n}\frac{b_{i}}{a_{i} + b_{i}} \geq 1. 
 \end{equation}
\end{theorem}

\begin{proof}
\Cref{thm:asl-one-booker} follows from \cref{thm:multi-odds} (proved further) for $m=1$. (Note that \cref{thm:asl-one-booker} is not used in the proof of \cref{thm:multi-odds}.)
\end{proof}
Note that, in practice, $\sum_{i=1}^{n}\frac{b_{i}}{a_{i} + b_{i}} $ is normally strictly greater than $1$, and
\begin{equation}
100\times\left(\sum_{i=1}^{n}\frac{b_{i}}{a_{i} + b_{i}} - 1\right)
\end{equation}
is called the \emph{over-round margin} \citep{2013:Emiliano,vlastakisbeating}.

Let's see an example of \cref{thm:asl-one-booker}.
\begin{example}\label{ex:single booky}
Suppose that a bookmaker provides betting odds $3/4$ for W, $13/5$ for D, and $16/5$ for L. As 
\begin{equation}
\frac{4}{3+4} + \frac{5}{13+5}+ \frac{5}{16+5} = 1.087 \geq 1,
\end{equation}
by \cref{thm:asl-one-booker}, the bookmaker avoids sure loss. Therefore, a customer cannot exploit these odds in order to make a sure gain.
\end{example}
Note that the condition for avoiding sure loss of $\mathcal{D}$ in \cref{thm:asl-one-booker} is exactly the same
as the condition for avoiding sure loss of $\underline{P}_{\overline{p}}$ in \cref{thm:asl:up-prob}. This condition is also equivalent to Proposition~4 in \citet{2015:Cortis}.

Next, we show that those odds can be modelled through an upper probability mass function:
\begin{lemma}\label{lem:gam-odds}
Let $ \Omega = \{\omega_1,\dots,\omega_n\}$, let $\omega_{i} \in \Omega$ and let $g$ be the corresponding gamble to the odds on $\omega_{i}$ defined as in \cref{eq:odds-gam2}, that is,
\begin{equation}\label{eq:7.4}
g_i(\omega)\coloneqq
\begin{cases}
-a_i \qquad \text{if }\omega = \omega_{i} \\ ~~ b_i\qquad \text{otherwise}, 
\end{cases}
\end{equation}
where $a_i$ and $b_i$ are non-negative.
If $p$ is a probability mass function, that is, if $\sum_{\omega \in \Omega} p(\omega) = 1$ and $p(\omega) \geq 0$ for all $\omega\in\Omega$, then 
\begin{equation}\label{eq:7.5}
\sum_{\omega \in \Omega} g_i(\omega) p(\omega) \geq 0 \qquad \Longleftrightarrow \qquad \frac{ b_i}{ a_i+ b_i} \geq p(\omega_{i}).
\end{equation} 
\end{lemma}

\begin{proof}
Suppose that $\sum_{\omega \in \Omega}  p(\omega) = 1$  and for all $i, p(\omega_i) \geq 0$, then
\begin{align}
\sum_{\omega \in \Omega} g_i(\omega) p(\omega) \geq 0
& \iff -a_ip(\omega_{i}) + b_i\sum_{\omega \neq \omega_{i}} p(\omega) \geq 0\\
& \iff -a_ip(\omega_{i}) + b_i(1- p(\omega_{i})) \geq 0 \\ \label{eq:odd>p(omega)}
& \iff \frac{ b_i}{ a_i + b_i} \geq p(\omega_{i}).
\end{align} 
\end{proof}
In order to avoid sure loss, the odds $a_i/b_i$ on $\omega_{i}$
must satisfy \cref{eq:odd>p(omega)} \citep[\S 3.3.3~(a)]{1991:walley} (see the proof of \cref{thm:multi-odds} for more detail). Therefore, the collection of these odds can be viewed as an upper probability mass function, that is, 
\begin{equation}\label{eq:set over_prob}
\forall i\in\{1, \dots, n\}\colon \overline{p}(\omega_{i})\coloneqq\frac{ b_{i}}{a_{i}+b_{i}}.
\end{equation}

\subsection{Betting with multiple bookmakers}\label{multi bookies}

In the market, there are many bookmakers. We are interested in whether a customer can exploit odds from different bookmakers in order to make a sure gain. To do so, we model betting odds from different bookmakers as a set of desirable gambles, and we check avoiding sure loss of this set. We recover the known result that it is optimal to pick maximal odds on each outcome \citep{vlastakisbeating}. As greater odds correspond to a higher payoff to a customer, a sensible strategy for him is to pick the greatest odds on each outcome.

\begin{theorem}\label{thm:multi-odds}
Let $ \Omega = \{\omega_1,\dots,\omega_n\}$. Suppose there are $m$ different bookmakers. For each $k\in\{1,\dots,m\}$, let $a_{ik}/b_{ik}$ be the betting odds on $\omega_{i}$ provided by bookmaker $k$.
 For each $i\in\{1,\dots,n\}$ and $k\in \{1,\dots,m\}$, let
\begin{equation}\label{eq:7.14}
 g_{ik}(\omega)\coloneqq
\begin{cases}
  -a_{ik} & \text{if } \omega = \omega_{i}\\
  b_{ik} & \text{otherwise}.
\end{cases}
\end{equation}
be the desirable gamble corresponding to the odds $a_{ik}/b_{ik}$.
Let $a_{i}^{*}/b_{i}^{*}$ be the maximal betting odds on outcome $\omega_{i}$, that is,
\begin{equation}\label{eq:7.13}
a_{i}^{*}/b_{i}^{*} \coloneqq \max_{k=1}^m \left\lbrace a_{ik}/b_{ik} \right\rbrace.
\end{equation}
 Then the set of desirable gambles $\mathcal{D}=\{g_{ik}\colon  i\in\{1,\dots,n\},\ k \in \{1,\dots,m\}\}$ avoids sure loss if and only if 
 \begin{equation}\label{eq:7.15}
 \sum_{i=1}^{n}\frac{ b_{i}^{*}}{ a_{i}^{*} + b_{i}^{*}} \geq 1.
 \end{equation}
\end{theorem}
\begin{proof}
See \cref{app:C}.
\end{proof}
\Cref{thm:multi-odds} tells us that to check avoiding sure loss of several bookmakers, we only need to consider the maximal odds on each outcome.
Let's see an example. 
\begin{example}
Suppose that in the market there are three bookmakers providing different odds for  outcomes W, D, and L as in \cref{table:multi-odds}.

\begin{center}
\begin{tabular}{|c|c|c|c|c|}\hline
\multirow{2}{*}{Outcomes} & \multicolumn{3}{c|}{Betting companies}&\multirow{2}{*}{Maximum odds} \\ \cline{2-4} 
 & River  & Mountain & Forest & \\ \hline
W & $4/5$ &  $17/20$ & $3/4$ &  $17/20$     \\
D & $13/5$ & $14/5$ & $13/5$ &  $14/5$   \\
L & $10/3$ &  $3$ & $16/5$ &  $10/3$     \\ \hline
\end{tabular}
\captionof{table}{Table of odds provided by three bookmakers}\label{table:multi-odds}
\end{center}

Let $\mathcal{D}$ be the set of desirable gambles corresponding to all of these odds.
Note that the maximal betting odds are $17/20$ for W, $14/5$ for D and $10/3$ for L.
As
\begin{equation}
\frac{20}{17+20} + \frac{5}{14+5}+ \frac{3}{10+3} = 1.034 \geq 1,
\end{equation}
by \cref{thm:multi-odds}, we conclude that $\mathcal{D}$ avoids sure loss. Therefore, a customer cannot exploit these odds to make a sure gain.
\end{example}

Consider a customer who is interested in odds provided by the three bookmakers as in \cref{table:multi-odds}. A sensible strategy to him is to pick the greatest odds on each outcome. However, this means that the customer will never choose any odds provided by Forest, because all of Forest's odds are less than the odds provided by other bookmakers. Therefore, to encourage customers to bet with them, Forest may offer free coupons to the customer under certain conditions. In the next section, we will look at these free coupons in more detail.

\section{Free coupons for betting}\label{sec:free bet}

A free coupon is a free stake that is given by a bookmaker to a customer who first bets with him. The free coupon can be spent on some betting odds that the customer wants to bet. In fact, the free coupon is not truly free, since the customer firstly has to bet on some odds before he claims the free coupon. Moreover, the bookmakers usually set some required conditions, for instance, a limit on the amount of free coupons that customers can claim, or a restriction of choices that customers can spend their free coupons. 
 
We were wondering whether customers can exploit those given odds and free coupons in order to find a strategy of betting that incurs a sure gain. If there is a possible way to do that, then we will find an algorithm that gives such a strategy.

For simplicity in this study, we set up standard requirements for claiming free coupons from the bookmakers as follows: 
\begin{enumerate}
\item Once the customer has placed his first bet, the bookmaker will give him a free coupon whose value is equal to the value of the bet that he placed.
\item The bookmaker sets the maximum value of the free coupon.
\item The free coupon only applies to the customer's first bet with the bookmaker.
\item The customer must spend his free coupon with the same bookmaker on other outcomes.
\item The customer must spend his free coupon on only a single outcome.
\end{enumerate}
Here is an example of claiming free coupons. 

\begin{example}\label{ex:free-coupon1}
Suppose that Forest has the following offer: a free coupon will be given to a customer who first bets with Forest, and the value of the coupon is equal to the value of the first bet that the customer placed.  

From \cref{table:multi-odds}, if James, who is a customer, has never bet with Forest and he decides to place \pounds$5$ on the odds $13/5$ of the outcome D, then he will play \pounds$5$ to Forest and he will claim a free coupon valued \pounds$5$. James can use the free coupon to bet on other outcomes with Forest.  
\end{example}

Once James receives a free coupon, he can spend his free coupons as in the next example.

\begin{example}\label{ex:free-coupon2}
Continuing from the previous example, James has his free coupon valued \pounds$5$ from Forest. 
Since James must spend his free coupon valued \pounds$5$ on only a single outcome, by \cref{lem:scale-odds}, we modify odds $3/4$ by multiplying them by $5/5$. Now all odds have the same denominator which is $5$.
\begin{center}
\begin{tabular}{l|c c c}
Outcomes       			&W &D & L\\
\hline
odds & $(\frac{3\cdot 5}{4})/5$ & $13/5$  & $16/5$
\end{tabular}
\captionof{table}{Table of modified odds}\label{table:odds-ex.4}
\end{center} 
If James spends his free coupon to bet on L and the true outcome is L, then Forest will lose \pounds$16$; otherwise Forest will lose nothing. On the other hand, if James spends the coupon to bet on W and the true outcome is W, then Forest will lose  \pounds$\frac{3\cdot 5}{4}$; otherwise Forest will lose nothing. A total payoff to Forest is summarised in \cref{table:payoff-ex.4}.

\begin{center}
\begin{tabular}{|c|c|c|c|}\hline
\multirow{2}{*}{Betting a free coupon on} & \multicolumn{3}{c|}{Outcomes} \\ \cline{2-4} 
& W  & D & L  \\ \hline
 L &  $0$ & $0$&  $-16$     \\
 W & $-\frac{3\cdot 5}{4}$ &$0$&$0$  \\ \hline
\end{tabular}
\captionof{table}{Table of total payoff}\label{table:payoff-ex.4}
\end{center}
\end{example}

Suppose that the customer first bets on an outcome $\omega_{i}$ with corresponding odds  $ a_{i}/ b_{i}$. The payoff to the bookmaker is represented as a gamble $g_{\omega_{i}}$ in the \cref{table:7.6.5}. Because this is his first bet, the customer receives a free coupon valued $b_{i}$, and he will spend this free coupon to bet on a single outcome. Suppose that he bets on $\omega_{j}$ with corresponding odds  $ a_{j}/ b_{j}$. As the denominators are not necessarily equal, we multiply odds  $ a_{j}/ b_{j}$ by $\frac{b_{i}}{b_{i}}$. The modified odds are $(\frac{a_j\cdot b_i}{b_j})/b_i$. 
Note that as the free coupon must be spent on other outcomes, $\omega_{j}$ cannot coincide with $\omega_{i}$.

If the true outcome is $\omega_{j}$, then the bookmaker will lose $ \frac{a_j\cdot b_i}{b_j}$. Otherwise the bookmaker will gain nothing. This payoff to the bookmaker is viewed as a gamble $\tilde{g}_{\omega_{j}}$ in the \cref{table:7.6.5}. As $g_{\omega_{i}}$ and $\tilde{g}_{\omega_{j}}$ are desirable to the bookmaker, by rationality axiom (D4),  $g_{\omega_{i}}+\tilde{g}_{\omega_{j}}$ is also desirable.

\begin{center}
\begin{tabular}{l|c c c}
Outcomes       			&$\omega_{i}$ &$\omega_{j}$ & others \\
\hline
$g_{\omega_{i}}$ & $-a_{i}$ & $b_i $  & $b_{i}$ \\
\hline
$\tilde{g}_{\omega_{j}}$ & $0$ & $-\frac{a_j\cdot b_i}{b_j} $  & $0$ \\
\hline
$g_{\omega_{i}}+\tilde{g}_{\omega_{j}}$ & $-a_{i}$ & $\frac{(b_j-a_j)b_i}{b_j} $  & $b_{i}$
\end{tabular}
\captionof{table}{Table of the first-free desirable gamble to the bookmaker}
\label{table:7.6.5}
\end{center} 
We denote  $g_{\omega_{i}\omega_{j}} \coloneqq g_{\omega_{i}}+\tilde{g}_{\omega_{j}}$ and call it the \textit{first-free} desirable gamble to the bookmaker. Note that $-g_{\omega_{i}\omega_{j}}$ is desirable to the customer. The customer can bet on other odds, but he will not get any free coupon from his additional bets. This is because the bookmaker gives him the free coupon only once.

Also note that in the actual market, there is usually more than one bookmaker offering a free coupon. Therefore, the customer can first bet with different bookmakers in order to obtain several free coupons. These can be viewed as a first-free desirable gamble combining from several first-free desirable gambles. In this study, we only consider the case that customer first bets and claims a free coupon from a single bookmaker. In this case, we face a combinatorial problem over all first-free desirable gambles.

We would like to check whether $\mathcal{D} \cup \{g_{\omega_{i}\omega_{j}}\}$ avoids sure loss or not. By \cref{thm:check-asl}, if $\mathcal{D}$  avoids sure loss, then $\mathcal{D} \cup \{g_{\omega_{i}\omega_{j}}\}$ avoids sure loss if and only if $\overline{E}(g_{\omega_{i}\omega_{j}}) \geq 0$. 
In the case that $\mathcal{D} \cup \{g_{\omega_{i}\omega_{j}}\}$ does not avoid sure loss, by \cref{thm:check-asl}, the bookmaker will lose at least  $|\overline{E}(g_{\omega_{i}\omega_{j}})|$ which is the customer's highest sure gain. Therefore, the customer can combine $g_{\omega_{i}\omega_{j}}$ with a non-negative combination of $g_i$ to obtain a sure gain $|\overline{E}(g_{\omega_{i}\omega_{j}})|$.

Let $f$ be any first-free desirable gamble to the bookmaker.
Before using the results in Section 2.3 to calculate the natural extension of $f$, we have to check whether $\mathcal{D} $ avoids sure loss. If $\underline{P}_{\overline{p}}$ does not avoid sure loss, then without a free coupon, there is a non-negative combination of gambles that the customer can exploit to make a sure gain. On the other hand, if $\underline{P}_{\overline{p}}$ avoids sure loss, then we can write $f$ in terms of its level sets and use \cref{cor:decompose} to calculate the natural extension of $f$. 

\begin{example}\label{ex:free-coupon3}
Let Forest provide betting odds on W, D, and L as in \cref{table:multi-odds}. By \cref{eq:set over_prob}, we have 
\begin{equation}
\overline{p}(W) = \frac{4}{7} \qquad \overline{p}(D) = \frac{5}{18} \qquad \overline{p}(L) =  \frac{5}{21}.
\end{equation}
Since $\overline{p}(W)+  \overline{p}(D)+\overline{p}(L)  \geq 1$,
$\underline{P}_{\overline{p}}$ avoids sure loss
by \cref{thm:asl:up-prob}.

Continuing from \cref{ex:free-coupon2}, suppose that James first bets on D and spends his free coupon to bet on L. Then, the first-free desirable gamble $g_{DL}$ to Forest is as follows:
\begin{center}
\begin{tabular}{l|c c c}
Outcomes       			& $W$ & $D$ & $L$ \\
\hline
$g_{D}$ & $5$ & $-13$  & $5$ \\
$g_{L}$ & $0$ & $0$ & $-16$ \\
$g_{DL}$ & $5$ & $-13$  & $-11$\\
\end{tabular}
\captionof{table}{Table of desirable gambles to Forest}
\label{table:free coupon-gambles}
\end{center} 
We decompose  $g_{DL} $ in terms of its level sets as
\begin{equation}\label{eq:decomp g_DL}
g_{DL} = -13I_{A_{0}} +2I_{A_{1}}+16I_{A_{2}}
\end{equation}
where $A_{0}= \{W,D,L\}$, $A_{1}= \{W,L\}$ and $A_{2}= \{W\}$.
By \cref{thm:ne:up-prob}, we have
\begin{align}
\overline{E}_{\overline{p}}(A_{0}) & = \min\{\overline{p}(W)+\overline{p}(D)+\overline{p}(L),1 \} = 1 \\
\overline{E}_{\overline{p}}(A_{1}) & = \min\{\overline{p}(W)+\overline{p}(L),1 \} = \frac{17}{21}\\
\overline{E}_{\overline{p}}(A_{2}) & =\min\{\overline{p}(W),1\} = \frac{4}{7}.
\end{align}
Substitute $\overline{E}_{\overline{p}}(A_{i})$, $i \in \{ 0,1,2\}$ into \cref{eq:decomp g_DL}. By \cref{cor:decompose}, we have
\begin{equation}
\overline{E}_{\overline{p}}(g_{DL}) = -13\overline{E}_{\overline{p}}(A_{0}) +2\overline{E}_{\overline{p}}(A_{1})+16\overline{E}_{\overline{p}}(A_{2}) = -\frac{47}{21}.
\end{equation}
As $\overline{E}_{\overline{p}}(g_{DL})= -\frac{47}{21}<0$, by \cref{thm:check-asl}, Forest does not avoid sure loss. Therefore, with the free coupon, James can make a sure gain.
\end{example}

How should James bet? Remember that $ \Omega = \{\omega_1,\dots,\omega_n\}$ and that $g_i$ is the corresponding gamble to the odds $a_i/b_i$ on $\omega_{i}$:
\begin{equation}\label{eq:gam-odd}
 g_{i}(\omega)=
\begin{cases}
-a_{i} \qquad \text{if } \omega = \omega_{i} \\ ~~ b_{i}\qquad \text{otherwise}.
\end{cases}
\end{equation}
Note that we can calculate $\overline{E}_{\overline{p}}(f)$, or $\overline{E}_{\mathcal{D}_{\underline{P}_{\overline{p}}}}(f)$, by \cref{def:E_P}, for any gamble $f$ by solving the following linear program:
\begin{align}\tag{Pa}\label{eq:Pa}
\linprogref{P} \qquad \min \quad & \alpha \\\label{eq:Pb}\tag{Pb}
    \text{subject to} \quad &  \begin{cases} \forall \omega \in \Omega\colon \alpha - \sum_{i=1}^{n} g_{i}(\omega)\lambda_{i} \geq  f(\omega) \\
 \forall i = 1,\dots,n \colon \lambda_i \geq 0,
 \end{cases}
\end{align}
where the optimal $\alpha$ gives $\overline{E}_{\overline{p}}(f)$.
If the optimal $\alpha$ is strictly negative, then
the optimal $\lambda_1$, \dots, $\lambda_n$ give a combination of bets for a customer to make a sure gain.
The dual of \linprogref{P} is 
\begin{align}\label{eq:Da}
\tag{Da}
   \linprogref{D} \qquad \max \quad & \sum_{\omega \in \Omega}f(\omega)p(\omega)  \\  \label{eq:Db1}\tag{Db1}
\text{subject to} \quad & \begin{cases} \forall g_i\colon  \sum_{\omega \in \Omega} g_i(\omega) p(\omega) \geq 0  \\ \forall \omega \colon  p(\omega) \geq 0 \\
 \sum_{\omega \in \Omega}p(\omega) = 1.\end{cases}
\end{align}
After applying \cref{lem:gam-odds}, the constraints in \cref{eq:Db1} become:
\begin{align}\label{eq:Db2}\tag{Db2}
\text{subject to} \quad & \begin{cases} \forall \omega \colon 0 \leq p(\omega) \leq \overline{p}(\omega) \\
 \sum_{\omega \in \Omega}p(\omega) = 1.\end{cases}
\end{align}
We see that the objective function \cref{eq:Da} is $E_{p}(f)$, the expectation of $f$ with respect to the probability mass function $p$. As the optimal value of \linprogref{D} is $\overline{E}_{\overline{p}}(f)$, if we can find a $p$ that satisfies the dual constraints \cref{eq:Db2} and $\overline{E}_{\overline{p}}(f) =E_{p}(f)$, then we have found an optimal solution of \linprogref{D}. 

We now first construct a $p$, by assigning as much mass as possible to the smallest level sets. Then, in \cref{thm:upE(f)=E(f)}, we prove that this $p$ satisfies \cref{eq:Db2} and  $\overline{E}_{\overline{p}}(f) =E_{p}(f)$. 

\begin{breakablealgorithm}
\caption{Construct an optimal solution $p$ of \linprogref{D}}
\label{alg:construct_p}
\begin{algorithmic}
\REQUIRE A gamble $f$, a set of outcomes $\Omega$.
\ENSURE An optimal solution $p$ of \linprogref{D}.
 \begin{enumerate}
\item Rewrite $f$ as
\begin{equation}
f = \sum_{i=0}^{m} \lambda_iA_i
\end{equation}
where $\Omega = A_{0} \supsetneq A_{1} \supsetneq\dots\supsetneq A_{m}\supsetneq \emptyset$ are the level sets of $f$ and $\lambda_{0} \in \mathbb{R}$, $\lambda_1,\dots,\lambda_{m} > 0$.
\item Order $\omega_1$, $\omega_2$, \dots, $\omega_n$ such that 
\begin{equation}\label{eq:label_omega}
\forall i \leq j\colon A_{\omega_i} \subseteq A_{\omega_j},
\end{equation} where $A_{\omega}$ is the smallest level set to which $\omega$ belongs, that is
\begin{equation}
A_{\omega} = \bigcap_{\substack{i=0\\ \omega \in A_i}}^{m}A_i.
\end{equation}
So, we start with those $\omega$ in $A_m$, then those in $A_{m-1}\setminus A_m$,
then those in $A_{m-2}\setminus A_{m-1}$,
and so on.
\item Let $k$ be the smallest index such that
\begin{equation}
\sum_{j=1}^{k}\overline{p}(\omega_j) \geq 1.
\end{equation} 
There is always such $k$ because $\underline{P}_{\overline{p}}$ avoids sure loss.
Define $p$ as follows:
\begin{align}\label{eq:p(omega)}
p(\omega_i) \coloneqq \begin{cases} \overline{p}(\omega_i) \quad  & \text{if } i< k  \\ 1- \sum_{j=1}^{i-1}\overline{p}(\omega_j)& \text{if } i = k \\
 0 \quad & \text{if } i > k.
\end{cases}
\end{align}
\end{enumerate}
\end{algorithmic}
\end{breakablealgorithm}

We then show that $p$ in \cref{eq:p(omega)} satisfies \cref{eq:Db2} and  $\overline{E}_{\overline{p}}(f) =E_{p}(f)$.
\begin{theorem}\label{thm:upE(f)=E(f)}
The probability mass function $p$ defined by \cref{eq:p(omega)} satisfies \cref{eq:Db2} and  $\overline{E}_{\overline{p}}(f) =E_{p}(f)$.
\end{theorem}
\begin{proof} 
Let $\Omega = \{\omega_1,\dots,\omega_n\}$ be ordered as in \cref{eq:label_omega}, and let $k$ be the smallest index such that $ \sum_{j=1}^{k}\overline{p}(\omega_j) \geq 1$.
By \cref{eq:p(omega)}, $\sum_{i=1}^{n}p(\omega_i) = 1$ and
\begin{equation}
  p(\omega_k)
  = 1- \sum_{j=1}^{k-1}\overline{p}(\omega_j)
  \leq \sum_{j=1}^{k}\overline{p}(\omega_j)- \sum_{j=1}^{k-1}\overline{p}(\omega_j)
  = \overline{p}(\omega_k), 
\end{equation}
so for all $i \in \{1,\dots, n\}$,  
$ 0 \leq p(\omega_i) \leq \overline{p}(\omega_i)$. Therefore, $p$ satisfies \cref{eq:Db2}. Next, we will show that for all level sets $A_i$, 
\begin{equation}\label{eq:min_p=sum_p}
\min\left\lbrace\sum_{\omega \in A_i}\overline{p}(\omega), 1 \right\rbrace = E_p(A_i).
\end{equation}
Remember that $A_{\omega_k}$ is the smallest level set that contains $\omega_k$. By \cref{eq:p(omega)}, for all $A_i \subsetneq  A_{\omega_k} $, we know that $p(\omega) = \overline{p}(\omega)$ for all $\omega \in A_i$, and so
\begin{equation}
\min\left\lbrace\sum_{\omega \in A_i}\overline{p}(\omega), 1 \right\rbrace = \sum_{\omega \in A_i}\overline{p}(\omega)  = \sum_{\omega \in A_i} p(\omega).
\end{equation}
For all $A_i \supseteq A_{\omega_k}$, we know that $ \sum_{\omega \in A_i} p(\omega) = 1$ and $ \sum_{\omega \in A_i} \overline{p}(\omega) \geq 1$, so
\begin{equation}
\min\left\lbrace\sum_{\omega \in A_i}\overline{p}(\omega), 1 \right\rbrace = 1  = \sum_{\omega \in A_i} p(\omega).
\end{equation} 
Hence, \cref{eq:min_p=sum_p} holds. Therefore,
\begin{align}
\overline{E}_{\overline{p}}(f) & = \sum_{i=0}^{m}\lambda_i\overline{E}(A_i) & \text{(by \cref{eq:cor:upE(f)})}\\
& = \sum_{i=0}^{m}\lambda_i\min\left\lbrace\sum_{\omega \in A_i}\overline{p}(\omega), 1 \right\rbrace  & \text{(by \cref{eq:ne-upper-prob})}\\
& = \sum_{i=0}^{m}\lambda_iE_p(A_i) & \text{(by \cref{eq:min_p=sum_p})} \\
& = E_p(f)
\end{align}
\end{proof}
To sum up, we can use \cref{eq:p(omega)} to construct  an optimal solution $p$ of \linprogref{D}.

We will use complementary slackness to find an optimal solution of the dual of \linprogref{D} \citep[p.~329]{1987:Winston}.
Note that, as \linprogref{D} has an optimal solution and the dual problem is bounded above, then by the strong duality theorem \cite[p.~71]{1995:Saigal}, an optimal solution of \linprogref{P} exists and achieves the same optimal value. In addition, a pair of solutions to \linprogref{P} and \linprogref{D} is optimal if, and only if, they satisfy the complementary slackness condition  \citep[p.~62]{1993:Fang:Puthenpura}. Specifically, in our case, the condition holds for any non-negative variable and its corresponding dual constraint \citep[p.~184, ll.~3--5]{2009:Griva:Nash:Sofer}.
More, precisely, let $p(\omega_1)$, \dots, $p(\omega_n)$ be any feasible solution of \linprogref{D}, and let $\alpha$, $\lambda_1$, \dots, $\lambda_n$ be any feasible solution of \linprogref{P}. 
Then, by complementary slackness,
these solutions are optimal if, and only if, for all $j\in\{1,\dots,n\}$, we have that
\begin{equation}
\left(\alpha - \sum_{i=1}^{n} g_{i}(\omega_j)\lambda_{i} - f(\omega_j)\right)p(\omega_j) =0\quad \text{and} \quad(\overline{p}(\omega_j) - p(\omega_j))\lambda_j = 0.
\end{equation}
This is equivalent to
\begin{enumerate}
\item if $p(\omega_j) >0$, then $ \alpha -\sum_{i=1}^{n} g_i(\omega_j)\lambda_{i} =  f(\omega_j)$, and 
\item if $p(\omega_j) < \overline{p}(\omega_j)$, then $\lambda_j =0$.
\end{enumerate}
So, if we have an optimal solution
$p(\omega_1)$, \dots, $p(\omega_n)$ of \linprogref{D}
and the optimal value $\alpha$,
then we can use these equations as a system of equalities in $\lambda_1$, \dots, $\lambda_n$.
Note that some solutions of this system may not satisfy feasibility, i.e. they may violate $\lambda_i\ge 0$.
However, all solutions of this system that satisfy
$\lambda_i\ge 0$ are guaranteed to be optimal solutions of \linprogref{P}.

How does this system of equalities look like?
Remember that $k$ was defined as the smallest
index such that $\sum_{j=1}^k \overline{p}(\omega_j)\ge 1$.
According to \cref{eq:p(omega)}, for all $j \in \{1,\dots, k-1\}$ we have that $p(\omega_j) > 0$, so we have the following equalities: for all $j \in \{ 1,\dots, k-1\}$,
\begin{equation}
\alpha -\sum_{i=1}^{n} g_i(\omega_j)\lambda_{i} =  f(\omega_j).
\end{equation}
For all $j \in \{k+1,\dots, n\}$ we have that $p(\omega_j)=0<\overline{p}(\omega_j)$, so $\lambda_j=0$ for all $j \in \{k+1,\dots, n\}$.
For $j=k$, if $p(\omega_k) < \overline{p}(\omega_k)$, then we can also set $\lambda_k=0$.
Otherwise, we know that $p(\omega_k)=\overline{p}(\omega_k)>0$
and so we can simply impose the same equality as
for $j\in \{1,\dots,k-1\}$. Concluding,
let $k'$ be the largest index $j$ for which $p(\omega_j)=\overline{p}(\omega_j)$.
Then as the optimal solution of \linprogref{P} exists, it can be found by solving the following system:
\begin{align}
\forall j\in\{1,\dots,k'\}\colon
&\alpha -\sum_{i=1}^{k'} g_i(\omega_j)\lambda_{i} =  f(\omega_j) \\
\forall j\in\{k'+1,\dots,n\}\colon
&\lambda_j=0
\end{align}
So, effectively, all we are left with is a system of $k'$ variables in $k'$ constraints.

Note that we can modify the odds to have the same denominator (all $b_i$ are equal), so it will be much easier to solve the new system. 

Finally, note that in the first-free coupon scenario, to make a sure gain, the customer has to bet on every outcome. This implies that the only coefficients $\lambda_i$ whose value can be zero are those corresponding to the gambles in the first-free gamble chosen by the customer. Hence, in that specific case, $k' \geq n-2$.

\begin{example}\label{ex:free-coupon4}
Continuing from \cref{ex:free-coupon3}, the corresponding linear programs to $\overline{E}(g_{DL})$ are as follows:
\begin{align}
\tag{P1a}
   \linprogref{P1} \qquad\min \quad & \alpha  \\ \label{eq:P1b}\tag{P1b}
\text{subject to} \quad & \begin{cases}\alpha +3\lambda_{W} - 5\lambda_{D} -5\lambda_{L} \geq 5\\ 
 \alpha  - 4\lambda_{W}+ 13\lambda_{D} -5\lambda_{L} \geq -13\\ 
 \alpha- 4\lambda_{W} - 5\lambda_{D} +16\lambda_{L} \geq -11\end{cases}\\ \tag{P1c}
 \text{and }\quad & \lambda_{W},\ \lambda_{D},\ \lambda_{L} \geq 0,\ \alpha \text{ free},
\end{align} 
\begin{align}
\tag{D1a}
   \linprogref{D1} \qquad \max \quad &  5p(W) - 13p(D) -11p(L)\\ \label{eq:dual-ex6}\tag{D1b}
\text{subject to}\quad &\begin{cases} 0 \leq p(W) \leq 4/7\\
0 \leq  p(D) \leq 5/18\\
0 \leq p(L) \leq 5/21\\
 p(W)+p(D)+p(L) = 1.\end{cases}
\end{align} 
By \cref{eq:label_omega}, we see that 
\begin{equation}
A_{W} \subseteq A_{L} \subseteq A_{D},
\end{equation}
so an optimal solution of \linprogref{D1} is as follows: 
\begin{equation}\label{eq:dual-sol}
p(W) = \frac{4}{7}, \quad p(L) = \frac{5}{21}, \quad p(D)= 1 - \left(\frac{4}{7}+ \frac{5}{21} \right) = \frac{4}{21}.
\end{equation}
As $p(W) =  \overline{p}(W)$ and $p(L) =  \overline{p}(L)$, whilst $p(D) < \overline{p}(D)$, by the complementary slackness, the optimal solution of \linprogref{P1} must have $\lambda_D = 0$ and solves the following system:
\begin{align}\tag{P1b1}\label{eq:P1b1}
\alpha +3\lambda_{W} -5\lambda_{L}& =5\\ \tag{P1b2}\label{eq:P1b2}
\alpha -4\lambda_{W} +16\lambda_{L}&= -11,
\end{align} 
where the value of $\alpha$ is $-\frac{47}{21}$. 
We solve this system and get an optimal solution: $\lambda_{W} = \frac{18}{7} $ and $\lambda_{L} = \frac{2}{21} $. 

A strategy for James to make a guaranteed gain is as follows. He first bets \pounds$5$ on D and claims a free coupon valued \pounds$5$ to bet on L. Next, he additionally bets \pounds$\frac{18}{7}$ on W and \pounds$\frac{2}{21}$ on D. He will make a sure gain of \pounds$\frac{47}{21}$ from Forest.
\end{example}

\section{Actual football betting odds}\label{sec:Euro-odds}
In this section, we will look at some actual odds in the market, and we will check whether and how a customer can exploit those odds and free coupons in order to make a sure gain.

Consider \cref{Table:all-euro-odds} which is in  \cref{app:D}. We list betting odds provided by 27 bookmakers on the winner of the European Football Championship 2016. From \cref{Table:all-euro-odds}, the maximum betting odds on each outcome are listed in \cref{table:euro-max}.
\begin{table}\addtolength{\tabcolsep}{-3pt}   
\begin{tabular}{|c c|c c|c c|}  
\hline
Country & Odds & Country &  Odds & Country &  Odds\\ 
\hline
France & 10/3 & Austria  & 45 & Czech Republic & 135\\
Germany & 23/5 & Poland  & 50 & Slovakia & 150\\
Spain & 5 & Switzerland &66 & Rep of Ireland & 170\\
England & 9 & Russia & 85 & Iceland & 180\\
Belgium  & 57/5 & Turkey & 94  & Romania & 275 \\
Italy & 91/5 &  Wales & 100& N Ireland &400\\
Portugal  & 20 & Ukraine & 100&  Hungary & 566\\ 
Croatia & 27 &Sweden &  104  & Albania & 531\\
\hline
\end{tabular}
\captionof{table}{Table of maximum betting odds for the European Football Championship 2016}\label{table:euro-max}
\end{table}
For all $i \in \{1,\dots, 24\}$, let $a_{i}^{*}/b_{i}^{*}$ be the maximal betting odds in \cref{table:euro-max}. Since $\sum_{i=1}^{24}\frac{ b_{i}^{*}}{ a_{i}^{*} + b_{i}^{*}} =  1.0349 \geq 1$, by \cref{thm:multi-odds}, the set of desirable gambles corresponding to the odds in \cref{Table:all-euro-odds} avoids sure loss. Therefore, there is no combination of bets which results in a sure gain.

Suppose that James is interested in betting with one of them, say Bet2. As he has never bet with Bet2 before, Bet2 will give him a free coupon on his first bet with them. With free coupons, we will check whether and how James can bet to make a guaranteed gain. Let $\mathcal{D}$ be a set of desirable gambles corresponding to the odds and let $g$ be any first-free desirable gamble to the company Bet2. We want to check whether $\mathcal{D} \cup \{g\}$ avoids sure loss or not. As there are 24 possible outcomes, the total number of different first-free desirable gambles with Bet2 is $24\times 23=552$. 

Suppose that James first bets on France and then spends his free coupon on Spain. So, the the first-free desirable gamble $g_{FG}$ is
\begin{center}
\begin{tabular}{l|c c c}
Outcomes       			&France & Spain & others \\
\hline
 $g_{F}$ & $-3$ & $1 $  & $1$ \\
 \hline
 $\tilde{g}_{S}$ & $0$ & $-5 $  & $0$ \\
\hline
 $g_{FS}$ & $-3$ & $-4$  & $1$ 
\end{tabular}
\captionof{table}{James' first-free gamble}\label{table:g_FG}
\end{center} 
where $F$ and $S$ denote France and Spain respectively. Again, we calculate $\overline{E}(g_{FS})$ by the Choquet integral. We decompose $g_{FS}$ in terms of its level sets as 
\begin{equation}\label{eq:g_FG}
g_{FS} = -4I_{A_{0}} +I_{A_{1}}+4I_{A_{2}}
\end{equation}
where $A_{0} = \Omega$, $A_{1} = \Omega\setminus\{S\} $
and $A_{2} = \Omega\setminus\{F,S\}$. By \cref{thm:ne:up-prob}, we have
\begin{equation}
\overline{E}(A_{0})= 1 \qquad \overline{E}(A_{1})= 0.9810\qquad \overline{E}(A_{2})= 0.7310.
\end{equation}
By \cref{cor:decompose}, we substitute $\overline{E}(A_{i})$, $i\in\{0, 1,2\}$ to \cref{eq:g_FG} and obtain 
\begin{equation}
\overline{E}(g_{FS}) = -4\overline{E}(A_{0}) +\overline{E}(A_{1})+4\overline{E}(A_{2}) = -0.0950.
\end{equation}
Therefore, $\mathcal{D} \cup \{g_{FS}\}$ does not avoid sure loss.

Among all possible first-free gambles, we find that there are three further gambles whose $\overline{E}$ is less than zero, namely ${\overline{E}(g_{FG}) =-0.2093}$, ${\overline{E}(g_{GF}) =-0.0117}$ and ${\overline{E}(g_{GS})=-0.0950}$, where $G$ denotes Germany. So, by \cref{thm:check-asl}, $\mathcal{D} \cup \{g\}$ does not avoid sure loss when $g \in \{g_{FS}, g_{FG}, g_{GF},g_{GS}  \}$; otherwise $\mathcal{D} \cup \{g\}$ avoids sure loss. Therefore, if
\begin{enumerate}
\item James first bets on France and then spends his free coupon to bet on either Spain or Germany, or 
\item James first bets on Germany and then spends his free coupon to bet on either France or Spain,
\end{enumerate}
then there is a combination of bets for him to bet in order to make a sure gain from Bet2. 

Consider the case where James first bets \pounds$1$ on France and claims his free coupon to bet on Spain. An optimal solution of the corresponding problem \linprogref{D} (the column $p(\omega_i)$ in \cref{table:bet2}) can be found through \cref{alg:construct_p}.
Then, we can find the optimal solution of the corresponding problem \linprogref{P} by using the optimal solution of \linprogref{D} with the complementary slackness condition. The optimal solution of \linprogref{P} is presented in a column $\lambda_i$ in \cref{table:bet2}. Therefore, if James additionally bets as in column $\lambda_i$, then he will make a sure gain of \pounds$0.095$ from Bet2.

\begin{table} 
\begin{tabular}{|c|c|c|c|c|c|}  
\hline
\multirow{2}{*}{Order $\omega_i$} & \multirow{2}{*}{Countries} & \multirow{2}{*}{Odds}& \multirow{2}{*}{$\overline{p}(\omega_i)$ }&\multicolumn{2}{c|}{Optimal solutions} \\ \cline{5-6} 
& & & & $p(\omega_i)$ &   $\lambda_i$ \\ \hline
1 &Germany & $4$ & $\frac{ 1}{5 }$& $\frac{ 1}{5 }$& $1$  \\
2& England & $9$ & $\frac{ 1}{ 10}$ &$\frac{ 1}{ 10}$& $0.5$  \\
3& Belgium & $10$ & $\frac{ 1}{ 11}$ & $\frac{ 1}{ 11}$&$\frac{5}{ 11}$  \\
4& Italy 	& $16$ & $\frac{ 1}{ 17}$&$\frac{ 1}{ 17}$& $\frac{5}{ 17}$  \\
5& Portugal & $18$ & $\frac{1 }{ 19}$&$\frac{1 }{ 19}$&$\frac{5}{ 19}$   \\
6& Croatia & $25$ &  $\frac{1 }{ 26}$&$\frac{1}{26}$& $\frac{5}{ 26}$   \\
7& Austria  & $40$ & $\frac{1}{41}$ &$\frac{1}{41}$& $\frac{5}{ 41}$  \\
8& Poland  & $50$  & $\frac{ 1}{51}$&$\frac{1}{51}$& $\frac{5}{ 51}$  \\
9& Switzerland &$40$ & $\frac{1}{41}$&$\frac{1}{41}$&$\frac{5}{41}$   \\
10& Russia  & $66$  & $\frac{1 }{ 67}$& $\frac{1}{67}$&$\frac{5}{ 67}$  \\
11& Turkey  & $80$ &$\frac{1 }{81 }$   &$\frac{1 }{81 }$& $\frac{5 }{81 }$  \\
12& Wales  & $80$  & $\frac{1 }{ 81}$  &$\frac{1 }{ 81}$&$\frac{5}{ 81}$   \\
13&  Ukraine & $66$& $\frac{1}{67}$&$\frac{1}{67}$&$\frac{5}{ 67}$   \\
14& Sweden  & $80$  & $\frac{ 1}{81 }$  &$\frac{ 1}{81 }$&$\frac{5}{81 }$   \\
15& Czech Republic &$100$&$\frac{1}{101}$& $\frac{1}{101}$&$\frac{5}{101}$   \\
16& Slovakia  & $100 $& $\frac{1 }{101 }$&$\frac{1}{101}$& $\frac{5}{101}$  \\
17& Rep of Ireland&$150 $&$\frac{1}{151}$&$\frac{1}{151}$& $\frac{5}{151}$  \\
18& Iceland  & $150 $  & $\frac{1}{151}$  &$\frac{1}{151}$&$\frac{5}{151}$   \\
19& Romania  & $100$ & $\frac{1 }{101}$&$\frac{1 }{101 }$&$\frac{5 }{101 }$   \\
20& N Ireland  & $250 $&$\frac{1}{251}$ &$\frac{1 }{251}$&$\frac{5 }{251}$   \\
21& Albania  &$250 $& $\frac{1}{251}$&$\frac{1 }{251 }$& $\frac{5 }{251 }$  \\
22& Hungary  & $250 $  &  $\frac{1}{251}$ & $\frac{5}{251}$&  $\frac{5}{251}$  \\
23& France & $3$ & $\frac{1}{4}$ & $\frac{1}{4}$  &  $\frac{1}{4}$ \\
24 & Spain 	& $5$ & $\frac{ 1}{6 }$  &  $\frac{586}{579}$   & $0$  \\
\hline
\end{tabular}
\captionof{table}{A summary of odds provided by Bet2, the upper probability mass function $\overline{p}(\omega_i)$, and optimal solution of \linprogref{D} and \linprogref{P}}\label{table:bet2}
\end{table}

\section{Conclusion}\label{sec:conclusion}

In this paper, we studied whether and how a customer can exploit given betting odds and free coupons in order to make a sure gain. Specifically, we viewed these odds and free coupons as a set of desirable gambles and checked whether such a set avoids sure loss or not via the natural extension. We showed that the set avoids sure loss if, and only if, the natural extension of the first-free gamble corresponding to the free coupon is non-negative. If the set does not avoid sure loss, then a combination of bets can be derived from the optimal solution of the corresponding linear programming problem.

We showed that for this specific problem, we can easily find the natural extension through the Choquet integral. In the case that the set does not avoid sure loss, we presented how to use the Choquet integral and the complementary slackness condition to directly obtain the desired combination of bets, without actually solving linear programming problems, but instead just solving a linear system of equalities. This technique can be applied to arbitrary problems involving upper probability mass functions.

To illustrate the results, we looked at some actual  betting odds on the winning of the European Football Championship 2016 in the market, and checked avoiding sure loss. We found that any sets of desirable gambles derived from those odds avoid sure loss. Having said that, with a free coupon, we identified sets of desirable gambles that no longer avoid sure loss. So, interestingly, in this case, when a free coupon is added, there was a combination of bets from which the customer could have made a sure gain.

\section*{Acknowledgements}

We would like to acknowledge support for this project from Development and Promotion of Science and Technology Talents Project (Royal Government of Thailand scholarship). We also thank the reviewers for their constructive comments.

\bibliographystyle{plainnat}
\bibliography{references}

\begin{thebibliography}{16}
\providecommand{\natexlab}[1]{#1}
\providecommand{\url}[1]{\texttt{#1}}
\expandafter\ifx\csname urlstyle\endcsname\relax
  \providecommand{\doi}[1]{doi: #1}\else
  \providecommand{\doi}{doi: \begingroup \urlstyle{rm}\Url}\fi

\bibitem[Cortis(2015)]{2015:Cortis}
Dominic Cortis.
\newblock Expected values and variances in bookmaker payouts: A theoretical
  approach towards setting limits on odds.
\newblock \emph{The Journal of Prediction Markets}, 9\penalty0 (1):\penalty0
  1--14, 2015.

\bibitem[Emiliano(December 2013)]{2013:Emiliano}
Colantonio Emiliano.
\newblock Betting markets: opportunities for many?
\newblock \emph{Annals of the University of Oradea, Economic Science Series},
  22\penalty0 (2):\penalty0 200--208, December 2013.

\bibitem[Fang and Puthenpura(1993)]{1993:Fang:Puthenpura}
Shu-Cherng Fang and Sarat Puthenpura.
\newblock \emph{Linear Optimization and Extensions: Theory and Algorithms}.
\newblock 1993.

\bibitem[Griva et~al.(2009)Griva, Nash, and Sofer]{2009:Griva:Nash:Sofer}
Igor Griva, Stephen~G. Nash, and Ariela Sofer.
\newblock \emph{Linear and Nonlinear Optimization Second edition}.
\newblock SIAM, Philadelphia, 2009.

\bibitem[Milliner et~al.(2009)Milliner, White, and
  Webber]{2009:Milliner:White:Webber}
I.~Milliner, P.~White, and D.~Webber.
\newblock A statistical development of fixed odds betting rules in soccer.
\newblock \emph{Journal of Gambling, Business and Economics}, 3\penalty0
  (1):\penalty0 89--99, 2009.

\bibitem[Nakharutai et~al.(2018)Nakharutai, Troffaes, and
  Caiado]{2018:Nakharutai:Troffaes:Caiado}
N.~Nakharutai, M.~C.~M. Troffaes, and C.~C.~S. Caiado.
\newblock Improved linear programming methods for checking avoiding sure loss.
\newblock \emph{International Journal of Approximate Reasoning}, 101:\penalty0
  293--310, October 2018.
\newblock \doi{10.1016/j.ijar.2018.07.013}.

\bibitem[Quaeghebeur et~al.(2017)Quaeghebeur, Wesseling, Beauxis-Aussalet,
  Piovesan, and Sterkenburg]{2017:Quaeghebeur:Wesseling:Beauxis-Aussalet}
Erik Quaeghebeur, Chris Wesseling, Emma Beauxis-Aussalet, Teresa Piovesan, and
  Tom Sterkenburg.
\newblock The {CWI} world cup competition: Eliciting sets of acceptable
  gambles.
\newblock In Alessandro Antonucci, Giorgio Corani, In{\'e}s Couso, and
  S{\'e}bastien Destercke, editors, \emph{Proceedings of the Tenth
  International Symposium on Imprecise Probability: Theories and Applications},
  volume~62 of \emph{Proceedings of Machine Learning Research}, pages 277--288.
  PMLR, July 2017.

\bibitem[Saigal(1995)]{1995:Saigal}
Romesh Saigal.
\newblock \emph{Linear programming : a modem integrated analysis}.
\newblock Springer Science+Business Media New York, 1995.

\bibitem[Schervish et~al.(1998)Schervish, Seidenfeld, and
  Kadane]{1998:Schervish:Seidenfeld:Kadane}
Mark~J. Schervish, Teddy Seidenfeld, and Joseph~B. Kadane.
\newblock \emph{Some Measures of Incoherence: How Not to Gamble if You Must}.
\newblock Technical Report No 660, Department of Statistics, Carnegie Mellon
  Univeristy, 1998.

\bibitem[Troffaes and
  de~Cooman(2014)]{2014:troffaes:decooman::lower:previsions}
Matthias C.~M. Troffaes and Gert de~Cooman.
\newblock \emph{Lower Previsions}.
\newblock Wiley Series in Probability and Statistics. Wiley, 2014.
\newblock ISBN 978-0-470-72377-7.
\newblock URL
  \url{http://eu.wiley.com/WileyCDA/WileyTitle/productCd-0470723777.html}.

\bibitem[Troffaes and Hable(2014)]{2014:troffaes:itip:computation}
Matthias C.~M. Troffaes and Robert Hable.
\newblock \emph{Introduction to Imprecise Probabilities}, chapter Computation,
  pages 329--337.
\newblock Wiley, 2014.
\newblock \doi{10.1002/9781118763117.ch16}.

\bibitem[Vlastakis et~al.(2006)Vlastakis, Dotsis, and
  Markellos]{vlastakisbeating}
Nikolaos Vlastakis, George Dotsis, and Raphael~N. Markellos.
\newblock Beating the odds: Arbitrage and wining strategies in the football
  betting market.
\newblock Universidad Complutense, Madrid, Spain, 2006.
\newblock URL
  \url{http://financedocbox.com/Stocks/71222948-2006-annual-conference-june-28-july-1-2006-universidad-complutense-madrid-spain.html}.

\bibitem[Walley(1991)]{1991:walley}
Peter Walley.
\newblock \emph{Statistical Reasoning with Imprecise Probabilities}.
\newblock Chapman and Hall, London, 1991.

\bibitem[Williams(1975)]{1975:williams:condprev}
Peter~M. Williams.
\newblock Notes on conditional previsions.
\newblock Technical report, School of Math. and Phys. Sci., Univ. of Sussex,
  1975.

\bibitem[Williams(2007)]{2007:williams:condprev}
Peter~M. Williams.
\newblock Notes on conditional previsions.
\newblock \emph{International Journal of Approximate Reasoning}, 44\penalty0
  (3):\penalty0 366--383, 2007.
\newblock \doi{10.1016/j.ijar.2006.07.019}.

\bibitem[Winston(1987)]{1987:Winston}
Wayne~L. Winston.
\newblock \emph{Operations research: Applications and algorithms}.
\newblock Duxbury press. Boston, 1987.

\end{thebibliography}

\appendix

\section{Proof of \cref{cor:decompose}}\label{app:A}

\begin{proof}
Since $A_{0} = \Omega$, we can write $f$ as
\begin{equation}
f =  \sum_{i=1}^{n} \lambda_{i} I_{A_{i}} + \lambda_{0}
\end{equation}
where $\lambda_{0} \in \mathbb{R}$, $\lambda_{1},\dots,\lambda_{n} > 0$ and $ A_{1} \supsetneq \dots\supsetneq A_{n}\supsetneq \emptyset$.
Then
\begin{align}
\begin{split}
-f & = -\sum_{i=1}^{n} \lambda_{i} (1-I_{A_{i}^{\mathsf{c}}})-\lambda_{0}\\
& = -\sum_{i=1}^{n} \lambda_{i}-\lambda_{0}+\sum_{i=1}^{n} \lambda_{i}I_{A_{i}^{\mathsf{c}}}.
\end{split}
\end{align}
Therefore, 
\begin{align}
\overline{E}_{\overline{p}}(f)& = -\underline{E}_{\overline{p}}(-f)\\ \label{eq:A.4}
& = -\left(-\sum_{i=1}^{n} \lambda_{i}-\lambda_{0}+\sum_{i=1}^{n}\lambda_{i}\underline{E}_{\overline{p}}(A_{i}^{\mathsf{c}})\right) \qquad \\
& = \lambda_{0}+ \sum_{i=1}^{n} \lambda_{i}(1-\underline{E}_{\overline{p}}(A_{i}^{\mathsf{c}})\\
& = \lambda_{0} + \sum_{i=1}^{n} \lambda_{i}\overline{E}_{\overline{p}}(A_{i}),
\end{align}
where \cref{eq:A.4} holds by constant additivity
and comonotone additivity \citep[p.~382, Prop.~C.5(v)\&(vii)]{2014:troffaes:decooman::lower:previsions}.
\end{proof}

\section{Proof of \cref{thm:check-asl}}
\label{app:B}

\begin{proof}
For the first part, suppose that $ f \in \mathcal{L}(\Omega)$ and  $\mathcal{D} = \{g_{i}\colon i \in \{1, \dots,n\}\}$ is a set of desirable gambles that avoids sure loss. We find that
\begin{align}\label{eq:thm:check-asl:1}
\begin{split}
\overline{E}_{\mathcal{D}}(f) & = \inf \left\lbrace \alpha \in \mathbb{R} : \alpha - f \geq \sum_{i=1}^{n} \lambda_{i}g_{i},  \lambda_{i} \geq 0 \right\rbrace \\ & = \min\left\lbrace \max_{\omega\in \Omega} \left( f(\omega)+ \sum_{i=1}^{n} \lambda_{i}g_{i}(\omega)  \right):  \lambda_{i} \geq 0\right\rbrace,
\end{split}
\end{align}
where the $\inf$ is actually a $\min$ because $\mathcal{D}$ is finite.
So, by \cref{lem:asl:condi},
\begin{equation}\label{eq:thm:check-asl:2}
\overline{E}_{\mathcal{D}}(f) \geq 0 \Longleftrightarrow \forall \lambda_i \geq 0, \max_{\omega\in \Omega} \left(\sum_{i=1}^{n} \lambda_i g_{i}(\omega) +f(\omega) \right) \geq 0.
\end{equation} 
For the second part, if $\mathcal{D}\cup \{f\}$ does not avoid sure loss, then $\overline{E}_{\mathcal{D}}(f) < 0$. So, by \cref{eq:thm:check-asl:1}, there exists an $\omega^*$ in $\Omega$ and some $\lambda_i \geq 0$ such that
\begin{equation}\label{eq:thm:check-asl:3}
  \overline{E}_{\mathcal{D}}(f)
  = f(\omega^*)+ \sum_{i=1}^{n} \lambda_{i}g_{i}(\omega^*)
  \geq f(\omega)+ \sum_{i=1}^{n} \lambda_{i}g_{i}(\omega),
  \ \forall \omega \in \Omega.
\end{equation} 
Hence there is a sure loss of at least $|\overline{E}_{\mathcal{D}}(f)|$.
\end{proof}

\section{Proof of \cref{thm:multi-odds}}\label{app:C}
\begin{proof}
Note that for each $i$ and $k$, we have
\begin{equation}\label{eq:7.16}
\frac{ a_{ik}}{b_{ik}} \leq \frac{ a_{i}^{*}}{\textstyle b_{i}^{*}}
\qquad \Longleftrightarrow \qquad
\frac{ b_{i}^{*}}{ a_{i}^{*} + b_{i}^{*}} \leq \frac{ b_{ik}}{ a_{ik} + b_{ik}}.
\end{equation}
So,
\begin{equation}\label{eq:7.17}
\frac{b_{i}^{*}}{ a_{i}^{*} + b_{i}^{*}} = \min_{k} \left\lbrace\frac{ b_{ik}}{a_{ik} + b_{ik}}\right\rbrace.
\end{equation} 
$(\Longrightarrow)$
Suppose the set of desirable gambles $\mathcal{D}$ avoids sure loss.
We will show that \cref{eq:7.15} holds.
As $\mathcal{D}$ avoids sure loss, the following system of linear inequalities:
\begin{align}\label{eq:7.18}
\forall i\colon & p(\omega_i) \geq  0 \\ \label{eq:7.19}
& \sum_{i=1}^{n} p(\omega_i) = 1\\ \label{eq:7.20}
\forall i, k\colon & \sum_{i=1}^{n} g_{ik}(\omega_i) p(\omega_i) \geq 0,
\end{align}
has a solution \citep[p.~175, ll.~10--13]{1991:walley}, say 
$p = (p(\omega_{1}),\dots,p(\omega_{n}))$.
By \cref{lem:gam-odds}, for each $i$ and $k$,
\begin{equation}\label{eq:7.21}
\frac{ b_{ik}}{ a_{ik} + b_{ik}} \geq p(\omega_{i}).
\end{equation}
Then, by \cref{eq:7.17} for each $i$,
\begin{equation}\label{eq:7.22}
\frac{b_{i}^{*}}{a_{i}^{*} + b_{i}^{*}} \geq p(\omega_{i}).
\end{equation}
Therefore, 
\begin{equation}\label{eq:7.23}
\sum_{i=1}^{n}\frac{ b_{i}^{*}}{ a_{i}^{*} + b_{i}^{*}} \geq \sum_{i=1}^{n} p(\omega_{i}) = 1. 
\end{equation}
$(\Longleftarrow)$ Suppose $
 \sum_{i=1}^{n}\frac{ b_{i}^{*}}{ a_{i}^{*} + b_{i}^{*}} \geq 1
$ holds.
Let 
\begin{equation}\label{eq:prob}
S= \sum_{i=1}^{n}\frac{ b_{i}^{*}}{ a_{i}^{*} + b_{i}^{*}} \quad \text{ and }\quad p(\omega_{i}) = \frac{ b_{i}^{*}}{ S (a_{i}^{*} + b_{i}^{*})}.
\end{equation}
If we show that $p$ is a feasible solution of \cref{eq:7.18,eq:7.19,eq:7.20}, then $\mathcal{D}$ avoids sure loss.
Note that by \cref{eq:prob}, $p(\omega_i)\geq 0$ for all $i$, $\sum_{i=1}^{n} p(\omega_i) = 1$ and with \cref{eq:7.17}, $\frac{ b_{ik}}{ a_{ik} + b_{ik}} \geq p(\omega_{i})$. So, by \cref{lem:gam-odds}, $ \sum_{i=1}^{n} g_{ik}(\omega) p(\omega_i) \geq 0$ holds for all $g_{ik} $. Therefore, $p$ is a feasible solution of \cref{eq:7.18,eq:7.19,eq:7.20} and by \citep[p.~175, ll.~10--13]{1991:walley}, $\mathcal{D}$ avoids sure loss.
\end{proof}

\begin{landscape}
\section{Betting odds on the winner of the European Football Championship 2016}\label{app:D}
\begin{table}[ht]
\centering
\begin{adjustbox}{width=1.5\textwidth}
\begin{tabular}{|c|c|c|c|c|c|c|c|c|c|c|c|c|c|c|c|c|c|c|c|c|c|c|c|c|c|c|c|} 
\hline
\multirow{2}{*}{Countries} & \multicolumn{27}{c|}{bookmakers}\\ \cline{2-28} 
 &  \rot{ Bet1}  &  \rot{ Bet2}  & \rot{ Bet3} &  \rot{ Bet4}&  \rot{ Bet5} &  \rot{ Bet6 } &  \rot{ Bet7} &  \rot{ Bet8} &  \rot{ Bet9} &  \rot{ Bet10} & \rot{ Bet11} & \rot{ Bet12} &  \rot{ Bet13} &  \rot{ Bet14} &  \rot{ Bet15} & \rot{ Bet16} & \rot{Bet17} &  \rot{ Bet18} &  \rot{ Bet19} &  \rot{ Bet20} & \rot{Bet21} & \rot{ Bet22 }&  \rot{ Bet23} &  \rot{ Bet24} & \rot{ Bet25} & \rot{ Bet26 }& \rot{ Bet27}\\ \hline
 France & 3 &  3 & 3 &  3   & 3  & 11/4  & 3 & 16/5   & 3 & 16/5 & 16/5&   3 &  3 & 3 & 16/5 & 3 & 10/3 & 16/5 & 16/5 &16/5 & 3 & 16/5 & 3 & 3 & 3 & 3 & 3  \\ \hline

Germany &4  & 4 & 9/2 & 4 & 9/2 & 4 & 4 & 9/2 & 10/3 & 9/2 & 9/2 & 9/2 & 4 & 7/2& 4 & 9/2 & 9/2 & 19/5 & 9/2 & 15/4 & 9/2 & 19/5 & 9/2 & 4 & 9/2 & 22/5  & 23/5 \\ \hline

Spain & 5 & 5 & 9/2 & 5 & 9/2 & 5 & 5 & 9/2 & 5 & 9/2 &5 & 9/2 & 5  & 5 & 5 & 5 & 9/2 & 5 & 9/2 & 5 & 9/2 & 5 & 5 & 24/5 & 24/5 & 5 &5  \\ \hline

England & 17/2 & 9 & 9 & 8 & 9 & 8 &9  & 8 & 9 & 9 &8 & 9 & 8  & 8 & 9 & 9 & 9 & 17/2 & 9 & 9 & 9 & 17/2 & 8 & 8 & 17/2 &43/5  & 9 \\ \hline

Belgium & 11 & 10 & 10 & 10 & 10 & 11 & 10 & 11 &10  & 10 & 10& 10 & 11  &10  & 11 & 11 & 11 & 9 & 10 & 9 & 10 & 9 & 10 & 10 & 54/5 & 53/5 & 57/5\\ \hline

Italy & 16 & 16 & 18 &16  & 18 & 16 & 16 & 16 & 16 & 16 &16 & 18 & 18  & 16 & 18 & 16 & 18 & 16 & 16 &14  & 16 & 16 & 14 & 17 & 18 & 89/5 & 91/5 \\ \hline

Portugal& 18 & 18 & 18 & 18 & 18 & 18 & 18 & 14 & 20 & 17 &18 & 18 & 14  & 18 & 12 & 18 & 20 & 15 & 17 & 18 &17  & 15 & 18 & 268/17 & 88/5 & 92/5 & 91/5 \\ \hline

Croatia& 25 & 25 & 22 & 25 & 22 & 25 & 25 & 25 & 20 & 25 &25 & 22 & 25  & 25 & 25 & 25 & 25 & 25 & 25 & 25 & 25 & 25 & 22 & 25 & 26 & 24 & 27\\ \hline

Austria& 40 & 40 & 33 & 33 & 33 & 40 & 40 & 40 & 33 & 40 &40 & 40 & 33  & 33 & 28 & 40 & 33 & 40 & 40 & 33 & 40 & 40 & 33 & 40 & 45 & 43 & 45 \\ \hline

Poland & 50 & 50 & 50 & 50 & 50 & 40 & 40 & 50 & 50 & 50 & 40& 50 & 50  & 50 & 50 & 40 & 50 & 45 & 50 & 40 & 50 & 45 & 50 & 50 & 47 & 48 & 50 \\ \hline

Switzerland& 66 & 40 & 66 & 50 & 66 & 66 & 50 & 50 & 66 & 66 & 50& 50 & 50  & 50 & 50 & 66 & 66 & 66 & 66 & 66 & 66 & 66 & 50 & 60 & 66 & 65 & 64\\ \hline

Russia& 66 & 66 & 80 & 66 & 80 & 80 & 66 & 66 & 66 & 80 & 66& 50 & 66  & 66 & 50 & 66 & 66 & 66 & 80 & 66 & 80 &66  & 50 & 66 & 85 & 84 & 79\\ \hline

Turkey& 80 & 80 & 80 & 80 & 80 & 80 & 80 & 66 & 80 & 80 & 66& 66 & 66  & 66 & 80 & 80 & 66 & 80 & 80 & 80 & 80 & 80 & 80 & 80 & 94 & 92 & 89 \\ \hline

Wales& 80 & 80 & 80 & 80 & 80 & 80 & 66 & 80 & 100 & 80 & 80& 66 & 66  & 66 & 100 & 66 & 66 & 80 & 80 & 80 & 80 & 80 & 60 & 80 & 89 & 81 & 89 \\ \hline

Ukraine&100  & 66 & 80 & 80 & 80 & 80 &80  & 66 & 80 & 80 &80 & 50 & 80  & 80 &50 & 50 & 80 & 90 & 80 & 100 &80  &90  & 80 & 100 & 94 & 86 & 89 \\ \hline

Sweden& 100 & 80 & 100 & 80 & 100 & 100 & 100 & 100 & 100 & 100 & 100& 100 &100   &80  &66  &100  & 100& 100 & 100 &100  &100  &100  & 80 & 100 & 104 & 90 & 99 \\ \hline

Czech Rep& 125 & 100 & 125 & 80 & 125 & 100 & 100 & 125 & 80 & 100 &100 & 125 &  66 & 100 & 100 & 125 & 100 & 100 & 100 & 100 & 100 & 100 & 100 &100  & 132 & 135 &99  \\ \hline

Slovakia& 150 & 100 & 150 & 150 & 150 & 150 & 150 & 150 & 100 & 100 & 150& 150 & 150  & 150 & 100 & 125 &125  & 187/2 & 100 & 150 & 100 & 100 & 125 & 150 & 142 & 143 &119  \\ \hline

Rep of Ireland& 150 & 150 & 150 & 150 & 150 & 125 & 150 & 100 & 150 & 150 &150 & 150 & 100  & 125 & 100 & 125 & 125 & 349/4 & 150 & 150 & 150 & 112 &  125& 150 & 170 & 156 &149  \\ \hline

Iceland& 100 & 150 & 100 & 100 & 100 & 100 & 100 & 150 & 80 & 100 &80 & 150 &  80 & 100 & 100 & 150 & 100 & 110 & 100 & 100 & 100 &110  & 60 & 100 & 180 & 179 & 149 \\ \hline

Romania& 200 & 100 & 150 & 125 & 150 & 200 & 200 & 125 & 150 & 260 & 150& 150 & 150  & 150 & 80 & 200 & 150 & 399/4 & 260 & 200 & 260 & 287/4 & 125 & 200 & 275 & 256 & 238 \\ \hline

N Ireland& 350 & 250 & 400&  400 & 400 & 350 & 350 & 300 & 300 & 400& 300 & 300  &250 &250 &300 & 250 & 400 & 359/4 & 400 & 300 & 400 &120& 350 & 400 & 389 & 377 & 376\\ \hline

Hungary& 350 & 250 & 400 & 200 & 400 & 350 & 350 & 400 & 350 & 400 & 250& 300 & 250  & 200 & 200 & 250 & 250 & 359/4 & 400 & 250 & 400 & 359/4 & 250 & 350 & 541 & 566 & 79 \\ \hline

Albania& 500 & 250 & 500 & 400 & 500 & 350 & 500 & 400 & 500 & 500 & 250& 200 & 300  & 250 & 300 & 400 & 500 & 363/4 & 500 & 300 & 500 & 177/4 & 400 & 500 & 531 & 513 & 495 \\ \hline
\end{tabular}
\end{adjustbox}
\caption{Table of betting odds on the winner of the European Football Championship 2016 where bookmaker names are modified. Collect data from www.oddschecker.com/football/euro-2016/winner on 13-06-2016.}
\label{Table:all-euro-odds}
\end{table} 
\end{landscape}

\appendix

\end{document}